\RequirePackage{fix-cm}
\documentclass[smallextended]{svjour3}       
\smartqed  
\usepackage{graphicx}
\usepackage{graphicx}
\usepackage{amsfonts}
\usepackage{amssymb,epic, graphicx}
\usepackage[cmex10]{amsmath}
\usepackage{array}
\usepackage{color}
\usepackage{mdwmath}
\usepackage{mdwtab}
\usepackage{eqparbox}
\usepackage[tight,footnotesize]{subfigure}
\usepackage{algorithm}
\usepackage{algorithmic}

\begin{document}

\title{The Weight Distributions of Two Classes of $p$-ary Cyclic Codes with Few Weights
}

\titlerunning{Weight Distributions of Cyclic Codes}        

\author{Shudi Yang      \and
        Zheng-An Yao    \and~
        Chang-An Zhao
        }


\institute{S.D. Yang \at
              Department of Mathematics,
Sun Yat-sen University, Guangzhou 510275 and School of Mathematical
Sciences, Qufu Normal University, Shandong 273165, P.R.China \\
              Tel.: +86-15602338023\\
                            \email{yangshd3@mail2.sysu.edu.cn}           
           \and
           Z.-A. Yao \at
               Department of Mathematics,
Sun Yat-sen University, Guangzhou 510275, P.R. China
        \and
        C.-A. Zhao  \at
         Department of Mathematics,
Sun Yat-sen University, Guangzhou 510275, P.R. China}

\date{Received: date / Accepted: date}

\maketitle

\begin{abstract}
Cyclic codes have attracted a lot of research interest for decades as
they have efficient encoding and decoding algorithms.
 In this paper, for an odd prime $p$,
the weight distributions of two classes of $p$-ary cyclic codes
are completely determined. We show that both
codes have at most five nonzero weights.

\end{abstract}

\vspace{1\baselineskip}
\noindent$\displaystyle \mathbf{Mathematics~~ Subject~~ Classification~~~} $    11T71$ \cdot$94B15

\section{Introduction}\label{sec:intro}

 Throughout this paper, let $p$ be an odd prime. Denote by $\mathbb{F}_p$ a
finite field with $p$ elements. An $[n, \kappa, l\,]$ linear code
$C$ over $\mathbb{F}_p$ is a $\kappa$-dimensional subspace of
$\mathbb{F}_p^n$ with minimum distance $l$. Moreover, the code is
cyclic if every codeword $(c_0,c_1,\cdots,c_{n-1})\in C $ whenever
$(c_{n-1},c_0,\cdots,c_{n-2})\in C $. Any cyclic code $C$ of length
$n$ over $\mathbb{F}_p$ can be viewed as an ideal of
$\mathbb{F}_p[x]/(x^n-1)$. Therefore, $C=\left\langle
g(x)\right\rangle $, where $g(x)$ is the monic polynomial of lowest degree
and divides $x^n-1$. Then $g(x)$ is called the
generator polynomial and $h(x)=(x^n-1)/g(x)$ is called the parity-check
polynomial \cite{macwilliams1977theory}.

Let $A_i$ denote the number of codewords with Hamming weight
$i$ in a linear code $C$ of length $n$. The weight enumerator of  $C$  is defined by
$$A_0+A_1x+A_2x^2+\cdots+A_nx^n,$$
where $A_0=1$. The sequence $(A_0,A_1,A_2,\cdots,A_n)$ is called the weight
distribution of the code $C$.

 Cyclic codes have found wide applications in
cryptography, error correction, association schemes and network
coding due to their efficient encoding and decoding algorithms.
However, there are still many open problems in coding theory (for details
see \cite{charpin1998open,Ding2014binary,macwilliams1977theory}).

It is an interesting subject to study the weight distribution of a
linear code. Firstly, the information of the error correcting
capability of a code is achieved from the weight distribution, i.e.,
the minimum distance $l$ is the minimum positive integer $i$ such
that $A_i>0$. Secondly, the weight distribution of a cyclic code is
closely related to the lower bound on the cardinality of a set of
nonintersecting linear codes, which can be applied to prove
the existence of resilient functions with high nonlinearity (see Theorem 4 of \cite{Johansson2003construction}).
 Finally, cyclic codes  with few weights have found
interesting applications in cryptography
\cite{carlet2005linear,YD2006}. Therefore, the weight distribution
is the major basis of computing the error probability of error
detection and correction, and it is the primary tool of researching
the structure of a code, improving the inner relationship of
codewords for finding a new good code. We refer the reader to \cite{Ding2013cyclotomic}
and \cite{Ding2014binary} given by Ding $et~al.$ for details on constructing optimal or almost
optimal cyclic codes in the sense that they meet some bounds on linear codes.

In recent years, much attention has been paid to evaluating the weight distribution of
cyclic codes though it is usually an extremely difficult problem. However, they are known
only in a few special cases. For example, the authors in \cite{ding2009weight,ding2013hamming,McEliece1975irre}
studied the weight distributions of irreducible cyclic codes.
For reducible cyclic codes, the authors in
\cite{feng2007value,li2014hamming,luo2008cyclic,luo2008weight,zhou2014class}
settled the weight distributions of cyclic codes whose duals have
two zeros. The authors of \cite{zeng2010weight,zheng2014weight,zheng2013weight,Zhou2013fiveweight} dealt with a few classes of cyclic codes whose duals have three zeros. As for cyclic codes whose duals have arbitrary zeros, see \cite{liYue2014weight} or \cite{yangjing2013arbitrary} for example.


Let $m$ and $k$ be two positive integers with $m>k$. For now on, we denote by $\alpha$
a primitive
element of $\mathbb{F}_{p^m}$. Let
$h_1(x)$ and $h_2(x)$ be the minimal polynomials of $\alpha^{-(p^k+1)}$
and $\alpha^{-1}$ over $\mathbb{F}_p$, respectively. Obviously,
$h_1(x)$ and $h_2(x)$ are pairwise distinct and
$\mathrm{deg}(h_2(x))=m$. Moreover, it can be easily shown that
$\mathrm{deg}(h_1(x))=m/2$ if $m=2k$ and $m$ otherwise.

Let $C_1$ and $C_2$ be two cyclic codes over $\mathbb{F}_p$ of length
$n=p^m-1$ with parity-check polynomials $h_1(x)h_2(x)$ and
$(x-1)h_1(x)$, respectively. Hence, the dimensions of $C_1$ and
$C_2$ over $\mathbb{F}_p$ are $3m/2$ and ${m}/{2}+1$, respectively, if
$m=2k$; and otherwise, the dimensions of $C_1$ and $C_2$ are $2m$ and $m+1$, respectively.


Let $d=\mathrm{gcd}(k,m)$ denote the greatest common
divisor of $k$ and $m$. Take $s={m}/{d}$. Note that the cyclic code $C_1$ was defined
 by Carlet, Ding and Yuan in \cite{carlet2005linear} and a tight lower
bound on the minimum distance was also determined. Later, the authors in \cite{yuan2006weight} established
the weight distribution of $C_1$ for odd $s$ (see also \cite{feng2007value,lichao2009covering}).
However, to the best of our knowledge,  there is
no information about the weight distribution of $C_1$ in the
case of even $s$.

In this paper, we explicitly determine the weight distribution of the
code $C_1$ for even $s$ and the weight distribution of the
code $C_2$, respectively. Furthermore, the
results show that both $C_1$ and $C_2$ are cyclic codes
with few weights. In fact, the number of nonzero weights of these
codes is no more than $five$. This means that the two classes
of cyclic codes may be of use in
cryptography \cite{mceliece1978public} and secret sharing schemes
\cite{carlet2005linear}.


The remainder of this paper is organized as follows. In Section
\ref{sec:Preli}, we introduce some definitions and results on
quadratic forms and exponential sums. Section \ref{sec:C1}
investigates the weight distribution of the code $C_1$ for even $s$. Section
\ref{sec:C2} studies the weight distribution of the code $C_2$.
Section \ref{sec:conclusion} concludes this paper and makes some
remarks on this topic.

\section{Preliminaries}\label{sec:Preli}

We follow the notations in Section \ref{sec:intro}. Let $q$ be a
power of $p$ and $t$ be a positive integer. By identifying
the finite field $\mathbb{F}_{q^t}$ with a $t$-dimensional vector
space $\mathbb{F}^t_{q}$ over $\mathbb{F}_{q}$, a function $f(x)$
from $\mathbb{F}_{q^t}$ to $\mathbb{F}_{q}$ can be regarded as a
$t$-variable polynomial over $\mathbb{F}_{q}$. The function $f(x)$ is called
a quadratic form if it can be written as a
homogeneous polynomial of degree two on  $\mathbb{F}^t_{q}$ as
follows:
$$f(x_1,x_2,\cdots,x_t)=\sum_{1\leqslant i \leqslant j\leqslant t}a_{ij}x_ix_j,~~a_{ij}\in \mathbb{F}_{q}.$$
Here we fix a basis of  $\mathbb{F}^t_{q}$ over  $\mathbb{F}_{q}$
and identify each $x\in \mathbb{F}_{q^t}$ with a vector
$(x_1,x_2,\cdots,x_t)\in\mathbb{F}^t_{q}$.
 The rank of the quadratic form $f(x)$, rank$(f)$, is defined as the
codimension of the $\mathbb{F}_{q}$-vector space
$$W=\{x\in \mathbb{F}_{q^t}|f(x+z)-f(x)-f(z)=0, ~~for~~all~~z\in \mathbb{F}_{q^t}\}.$$
Then $|W|=q^{t-\mathrm{rank}(f)}$.

For a quadratic form $f(x)$ with $t$ variables over $\mathbb{F}_q$,
there exists a symmetric matrix $A$ of order $t$ over $\mathbb{F}_q$
such that $f(x)=XAX'$, where $X=(x_1,x_2,\cdots,x_t)\in
\mathbb{F}^t_q$ and $X'$ denotes the transpose of $X$. It is known
that there exists a nonsingular matrix $B$ over $\mathbb{F}_q$ such
that $BAB'$ is a diagonal matrix. Making a nonsingular linear
substitution $X=YB$ with $Y=(y_1,y_2,\cdots,y_t)\in \mathbb{F}^t_q$,
we have
$$f(x)=Y(BAB')Y'=\sum^r_{i=1}a_iy^2_i,\,\,\,a_i\in \mathbb{F}^*_q,$$
where $r$ is the rank of $f(x)$. The determinant $\mathrm{det}(f)$
of $f(x)$ is defined to be the determinant of $A$, and $f(x)$ is said to be
nondegenerate if $\mathrm{det}(f)\neq0$.

The lemmas introduced below will turn out to be of use in the
sequel.

\begin{lemma}(See Theorems 5.15 and 5.33 of \cite{lidl1983finite})\label{lm: solu of simple quadra form}
Let $\mathbb{F}_{p^t}$ be a finite field with  $p^t$ elements and
$\eta_t$ be the multiplicative quadratic character of
$\mathbb{F}_{p^t}$. For $a\in\mathbb{F}^*_{p^t}$,
\begin{eqnarray*}
\sum_{x\in
\mathbb{F}_{p^t}}\zeta^{\mathrm{Tr}^t_1(ax^2)}_p=\eta_t(a)(-1)^{t-1}(\sqrt{-1})^{\frac{t}{4}(p-1)^2}p^{\frac{t}{2}},
\end{eqnarray*}where $\zeta_p=e^{2\pi\sqrt{-1}/p}$ and $\mathrm{Tr}^t_1$ is a trace function from
$\mathbb{F}_{p^t}$ to $\mathbb{F}_{p}$ defined by
$$\mathrm{Tr}^t_1(x)=\sum^{t-1}_{i=0}x^{p^i},~~x\in
\mathbb{F}_{p^t}.$$
\end{lemma}

\begin{lemma}(See Theorems 6.26 and 6.27 of \cite{lidl1983finite})\label{lm: solution of quadra form}
Let $f$ be a nondegenerate quadratic form over $\mathbb{F}_q$,
$q=p^t$ for odd prime $p$, in $l$ variables. Define a function
$\upsilon(\cdot)$ over $\mathbb{F}_q$ by $\upsilon(0)=q-1$ and
$\upsilon(\rho)=-1$ for $\rho\in\mathbb{F}^*_q$. Then for
$b\in\mathbb{F}_q$ the number of solutions of the equation
$f(x_1,\cdots,x_l)=b$ is
\begin{eqnarray*}
\left\{\begin{array}{lll}q^{l-1}+\upsilon(b)q^{\frac{l-2}{2}}\eta_t\left((-1)^\frac{l}{2}\mathrm{det}(f)\right),
&&if~~ l~~ is~~ even,\\
q^{l-1}+q^{\frac{l-1}{2}}\eta_t\left((-1)^\frac{l-1}{2}b~
\mathrm{det}(f)\right), &&if~~ l~~ is
~~odd,\\
\end{array}
\right.
\end{eqnarray*}
where $\eta_t$ is the quadratic character of $\mathbb{F}_q$.
\end{lemma}
%

For convenience, we abbreviate the trace function $\mathrm{Tr}^m_1$
as $\mathrm{Tr}$ in the sequel. We will require the following lemma whose
proof can be found in \cite{coulter1998explicit,draper2007explicit,yu2014weight}.

\begin{lemma}\label{lm: exponentialsums}
Let $S(a)=\sum_{x\in
\mathbb{F}_{p^m}}\zeta_p^{\mathrm{Tr}(ax^{p^k+1})}$ and $d=\mathrm{gcd}(k,m)$.
Let $\upsilon_2(\cdot)$ denote
the 2-adic order function.
Then $Q(x)=\mathrm{Tr}(ax^{p^k+1})$ is a quadratic form and for any
$a\in \mathbb{F}^*_{p^m}$,\\
 \textcircled{1} If
$\upsilon_2(m)\leqslant \upsilon_2(k)$, then
 $\mathrm{rank}(Q(x))=m$ and
\begin{eqnarray*}
S(a)=\left\{\begin{array}{lll}~~\sqrt{(-1)^{\frac{p^d-1}{2}}}~p^{\frac{m}{2}},
&&\frac{p^m-1}{2}~~times,\\
-\sqrt{(-1)^{\frac{p^d-1}{2}}}~p^{\frac{m}{2}},
&&\frac{p^m-1}{2}~~times.\\
\end{array}
\right.
\end{eqnarray*}
\textcircled{2} If $\upsilon_2(m)=\upsilon_2(k)+1$, then
$\mathrm{rank}(Q(x))=m$ or $m-2d$ and
\begin{eqnarray*}
S(a)=\left\{\begin{array}{lll}-p^{\frac{m}{2}},
&&~~~\frac{p^d(p^m-1)}{p^d+1}~~times,\\
~~p^{\frac{m}{2}+d}, &&~~~\frac{p^m-1}{p^d+1}~~~~~~times.\\
\end{array}
\right.
\end{eqnarray*}
\textcircled{3} If $\upsilon_2(m)>\upsilon_2(k)+1$, then
$\mathrm{rank}(Q(x))=m$ or $m-2d$ and
\begin{eqnarray*}
S(a)=\left\{\begin{array}{lll}~~p^{\frac{m}{2}},
&&~~\frac{p^d(p^m-1)}{p^d+1}~~times,\\
-p^{\frac{m}{2}+d}, &&~~\frac{p^m-1}{p^d+1}~~~~~~times.\\
\end{array}
\right.
\end{eqnarray*}
\end{lemma}
\begin{remark}
The value of $S(a)$ and its frequency can be easily obtained from
Corollary 7.6 of \cite{draper2007explicit} and the rank of $Q(x)$
can be deduced immediately from the value of $S(a)$. We mention that
Lemma \ref{lm: exponentialsums} plays an important role in
calculating the weight distributions of the cyclic codes $C_1$ and
$C_2$ in the sequel.

\end{remark}

%
%
%
%
%
For later use, we define
\begin{eqnarray}
R_i=\{a\in \mathbb{F}^*_{p^m}\big|~ \mathrm{rank}(Q(x))=m-2di\},
~~i\in \{0,1\}.
\end{eqnarray}

From Lemma \ref{lm: exponentialsums}, for $\upsilon_2(m)\leqslant
\upsilon_2(k)$, we have
\begin{eqnarray*}
S(a)=\sqrt{(-1)^{\frac{p^d-1}{2}}}~\theta_0p^{\frac{m}{2}},
~~\theta_0\in \{\pm1\},
\end{eqnarray*}
and for $\upsilon_2(m)\geqslant \upsilon_2(k)+1$ with $i\in
\{0,1\}$,
\begin{eqnarray*}
S(a)=\theta_ip^{\frac{m+2di}{2}}, ~~\theta_i\in \{\pm1\}.
\end{eqnarray*}
Two subsets $R_{i,j}$ of $R_i$ for $i\in \{0,1\}$ are defined as
\begin{eqnarray}\label{Rij}
 R_{i,j}=\{a\in
R_i\big|~\theta_i=j\},~~ j=\pm1.
\end{eqnarray}
Then, the value of each $|R_{i}|$ and $|R_{i,j}|$ can be computed by  Lemma \ref{lm: exponentialsums}.

Let $r=\mathrm{rank}(Q(x))$. By making a nonlinear substitution to
$Q(x)$ and using Lemma \ref{lm: solu of simple quadra form} we have
\begin{eqnarray}\label{eq:S(a)}
S(a)&=&\sum_{x\in
\mathbb{F}_{p^m}}\zeta_p^{\mathrm{Tr}(ax^{p^k+1})}
=\sum_{x_1,\cdots,x_m\in
\mathbb{F}_{p}}\zeta_p^{a_1x^2_1+\cdots+a_rx^2_r}\nonumber \\
&=&\eta\left(\prod^r_{i=1}a_i\right)(\sqrt{-1})^{\frac{r}{4}(p-1)^2}p^{\frac{r}{2}}p^{m-r}\nonumber \\
&=&\eta\left(\prod^r_{i=1}a_i\right)(\sqrt{-1})^{\frac{r}{4}(p-1)^2}p^{m-\frac{r}{2}},
\end{eqnarray}
where $a_i\in \mathbb{F}^*_p$ for $i=1,\cdots,r$ and $\eta$ is the
quadratic character over $\mathbb{F}_p$.

In the sequel, we define $\Delta_i=\prod_{j=1}^{m-2di}a_j$ for $i\in
\{0,1\}$. The following property will be needed to determine the
weight distribution of cyclic codes.
\begin{lemma}\label{lm:type of qua form}
With notations as before. For $i\in \{0,1\}$ and $j=\pm1$, we have
\begin{eqnarray}\label{eq:eta2}
\eta\left((-1)^{[\frac{m-2di}{2}]}\Delta_i\right)=j~~ occurring~~
|R_{i,j}|~~ times,
\end{eqnarray}
where $[x]$ denotes the largest integer that is less than or equal
to $x$.
\end{lemma}
\begin{proof}
We only give the proof of the case that $\upsilon_2(m)\leqslant
\upsilon_2(k)$ since the other cases can be proved in a similar way.

We assume that $\upsilon_2(m)\leqslant \upsilon_2(k)$ for the rest of the
proof. Thus, we only need to prove the desired conclusion in the case of
$i=0$ since $r=m$. The discussion in this case is divided
into the following subcases.

 If $\upsilon_2(m)\geqslant1$, then
\begin{eqnarray*}
\eta\left((-1)^{[\frac{m}{2}]}\Delta_0\right)=\eta\left((-1)^{\frac{m}{2}}\right)\eta(\Delta_0)
=(\sqrt{-1})^{\frac{m}{4}(p-1)^2}\eta(\Delta_0),
\end{eqnarray*}
which is equal to the coefficient of $p^{m\!-\!\frac{r}{2}}$ in Equation
\eqref{eq:S(a)} for $r=m$. Since $\sqrt{(-1)^{\frac{p^d-1}{2}}}=1$, the desired assertion holds for this subcase by Lemma \ref{lm: exponentialsums}.

 If $\upsilon_2(m)=0$, then
 \begin{eqnarray}\label{eq:1 or -1}
 (-1)^{[\frac{m}{2}]}=(-1)^{\frac{m-1}{2}}=\left\{\begin{array}{lll}~~1,&&~~if~~
 m\equiv1\mod4,\\
 -1,&&~~if~~
 m\equiv3\mod4.\\
 \end{array}
\right.
\end{eqnarray}

Recall that $p$ is an odd prime. If $p\equiv1\mod4$, then $-1$ is a quadratic
residue over $\mathbb{F}_p$. Therefore,
\begin{eqnarray*}
\eta\left((-1)^{[\frac{m}{2}]}\Delta_0\right)
=\eta(\Delta_0)=(\sqrt{-1})^{\frac{m}{4}(p-1)^2}\eta(\Delta_0),
\end{eqnarray*}
which is also equal to the coefficient of $p^{\frac{m}{2}}$ in Equation
\eqref{eq:S(a)}. Note that $\sqrt{(-1)^{\frac{p^d-1}{2}}}=1$. Hence, the
desired assertion holds for this subcase.

If $p\equiv3\mod4$, then $-1$ is a quadratic nonresidue over
$\mathbb{F}_p$. By \eqref{eq:1 or -1}, we have
\begin{eqnarray*}
\eta\left((-1)^{[\frac{m}{2}]}\Delta_0\right)
=\left\{\begin{array}{lll}~~\eta(\Delta_0),&&~~if~~
 m\equiv1\mod4,\\
 -\eta(\Delta_0),&&~~if~~
 m\equiv3\mod4.\\
 \end{array}
\right.
\end{eqnarray*}
Note that $(\sqrt{-1})^{\frac{m}{4}(p-1)^2}$ equals to $\sqrt{-1}$ if $ m\equiv1\mod4$, and $-\sqrt{-1}$ if $ m\equiv3\mod4$. This implies that $\eta\left((-1)^{[\frac{m}{2}]}\Delta_0\right)$ is
equal to the coefficient of $\sqrt{-1}p^{\frac{m}{2}}$. Since $\sqrt{(-1)^{\frac{p^d-1}{2}}}=\sqrt{-1}$,
 the
desired assertion holds for this subcase.   \hfill\space$\qed$
\end{proof}

\section{The weight distribution of the
code $C_1$}\label{sec:C1}

We now focus on the weight distribution of the code $C_1$ as
described in Section \ref{sec:intro}. It follows from Delsarte's
Theorem \cite{delsarte1975subfield} that
$$C_1=\{\mathsf{c}_1(a,b):a,b\in \mathbb{F}_{p^m}\},$$ where $\mathsf{c}_1(a,b)=(\text{Tr}(ax^{p^k+1}+bx))_{x\in
\mathbb{F}^*_{p^m}}$.

Let $N_{a,b}(0)$ be the number of solutions $x\in \mathbb{F}_{p^m}$
of the equation
\begin{eqnarray}\label{N1:0}
 \text{Tr}(ax^{p^k+1}+bx)=0,
\end{eqnarray}
as $(a,b)$ runs through $\mathbb{F}^2_{p^m}$. For a given basis
$\{\alpha_1,\alpha_2,\dots,\alpha_m\}$ of $\mathbb{F}_{p^m}$ over
$\mathbb{F}_{p}$, each $x\in \mathbb{F}_{p^m}$ can be uniquely
expressed as $x=\sum^m_{i=1}x_i\alpha_i$ with
$x_i\in\mathbb{F}_{p}$. Therefore, by making a nonsingular linear
substitution as introduced in Section \ref{sec:Preli}, Equation \eqref{N1:0}
becomes
\begin{eqnarray}\label{N1:0 2}
 \sum^m_{i=1}a_ix^2_i+\sum^m_{i=1}b_ix_i=0,
\end{eqnarray} where $a_i,b_i\in\mathbb{F}_{p}$. Hence, $N_{a,b}(0)$ also represents the number
of $(x_1,x_2,\dots,x_m)\in \mathbb{F}^m_{p}$ satisfying \eqref{N1:0 2}.

Recall that $d=\mathrm{gcd}(k,m)$ and $s={m}/{d}$. Note that $s$ is odd if and only
if $\upsilon_2(m)\leqslant \upsilon_2(k)$, and $s$ is even
if and only if $\upsilon_2(m)\geqslant \upsilon_2(k)+1$. For the case of
$s$ being odd, the references \cite{feng2007value,lichao2009covering,yuan2006weight} have given the weight
distribution of $C_1$ independently. In the following, we establish the weight distribution of $C_1$ for
even $s$.
\begin{theorem}\label{thm:code a bx} With notation given
before. If $\upsilon_2(m)\geqslant \upsilon_2(k)+1$ and $m\neq2k$, then $C_1$ is a cyclic code over $\mathbb{F}_p$
with parameters $[p^m-1,2m]$ and \\
\textcircled{1} If $\upsilon_2(m)=\upsilon_2(k)+1$, the weight
distribution of $C_1$ is given as follows:
\begin{eqnarray}\label{W1:31 m not equal 2k}
\left\{\begin{array}{l}A_0=1,\\A_{(p-1)p^{m-1}}=(p^m-1)(1+p^{m-d}-p^{m-2d}), \\
  A_{(p-1)(p^{m-1}+p^{\frac{m-2}{2}})}=(p^{m-1}-(p-1)p^{\frac{m-2}{2}})\frac{p^d(p^m-1)}{p^d+1},  \\
  A_{(p-1)p^{m-1}-p^{\frac{m-2}{2}}}=(p-1)(p^{m-1}+p^{\frac{m-2}{2}})\frac{p^d(p^m-1)}{p^d+1},  \\
 A_{(p-1)(p^{m-1}-p^{\frac{m+2d-2}{2}})}=(p^{m-2d-1}+(p-1)p^{\frac{m-2d-2}{2}})\frac{p^m-1}{p^d+1},   \\
  A_{(p-1)p^{m-1}+p^{\frac{m+2d-2}{2}}}=(p-1)(p^{m-2d-1}-p^{\frac{m-2d-2}{2}})\frac{p^m-1}{p^d+1}.\\
\end{array}
\right.
\end{eqnarray}
\textcircled{2} If $\upsilon_2(m)>\upsilon_2(k)+1$, the weight
distribution of $C_1$ is given as follows:
\begin{eqnarray}\label{W1:4 v2(m)>v2(k)+1}
\left\{\begin{array}{l}A_0=1,\\A_{(p-1)p^{m-1}}=(p^m-1)(1+p^{m-d}-p^{m-2d}), \\
   A_{(p-1)(p^{m-1}-p^{\frac{m-2}{2}})}=(p^{m-1}+(p-1)p^{\frac{m-2}{2}})\frac{p^d(p^m-1)}{p^d+1},  \\
   A_{(p-1)p^{m-1}+p^{\frac{m-2}{2}}}=(p-1)(p^{m-1}-p^{\frac{m-2}{2}})\frac{p^d(p^m-1)}{p^d+1},  \\
   A_{(p-1)(p^{m-1}+p^{\frac{m+2d-2}{2}})}=(p^{m-2d-1}-(p-1)p^{\frac{m-2d-2}{2}})\frac{p^m-1}{p^d+1},   \\
   A_{(p-1)p^{m-1}-p^{\frac{m+2d-2}{2}}}=(p-1)(p^{m-2d-1}+p^{\frac{m-2d-2}{2}})\frac{p^m-1}{p^d+1}.
\end{array}
\right.
\end{eqnarray}
\end{theorem}

\begin{proof}
From the definition of $C_1$, we know that $C_1$ has length $p^m-1$ and dimension $2m$. The Hamming weight of every codeword
$\mathsf{c}_1(a,b)$ can be determined by
\begin{eqnarray}\label{W1:1234}
\omega t(\mathsf{c}_1(a,b))&=&p^m-1-\#\{x\in
\mathbb{F}^*_{p^m}\big|~\text{Tr}(ax^{p^k+1}+bx)=0\}\nonumber\\
&=&p^m-\#\{x\in
\mathbb{F}_{p^m}\big|~\text{Tr}(ax^{p^k+1}+bx)=0\}\nonumber\\
&=&p^m-N_{a,b}(0).
\end{eqnarray}
It suffices to study the value distribution of $N_{a,b}(0)$. So,
we calculate the weight distribution of the code $C_1$ in the following
cases.

\textcircled{1} $\upsilon_2(m)=\upsilon_2(k)+1$ and $m\neq2k$.

 The value of $N_{a,b}(0)$
 will be calculated according to the choice of the parameter $a$.

%

\emph{Case 1:} $a=0$. In this case, if $b=0$ then $N_{a,b}(0)=p^m$
occurring only once, and if $b \neq 0$ then
$N_{a,b}(0)=p^{m-1}$ occurring $p^m-1$ times.

\emph{Case 2:} $a\in R_0$. In this case, rank$(Q(x))=m$ and consequently every coefficient $a_i$
in \eqref{N1:0 2} is nonzero.

For $1\leqslant i\leqslant m,$ let $x_i=y_i-\frac{b_i}{2a_i},$ then
\eqref{N1:0 2} is equivalent to $\sum^m_{i=1}
a_iy^2_i=\sum ^m_{i=1}\frac{b^2_i}{4a_i}$. It then follows from Lemma
\ref{lm: solution of quadra form} that
\begin{eqnarray}\label{eq13:N0}N_{a,b}(0)=p^{m-1}+\upsilon\left(\sum^m_{i=1}\frac{b^2_i}{4a_i}\right)p^{\frac{m-2}{2}}\eta((-1)^{\frac
m2}\Delta_0).
\end{eqnarray}

Notice that the tuple $(b_1,\dots,b_m)$ runs through
$\mathbb{F}^m_{p}$ as $b$ runs through $\mathbb{F}_{p^m}$. We can
regard $\sum ^m_{i=1}\frac{b^2_i}{4a_i}$ as a quadratic form
in $m$ variables $b_i$ for $1\leqslant i\leqslant m$.
Again by Lemma \ref{lm: solution of quadra form}, as $b$ runs through $\mathbb{F}_{p^m}$, we obtain
\begin{equation}\label{eq13:bi}
    \sum^m_{i=1}\frac{b^2_i}{4a_i}=\beta~~ occurring~~
p^{m-1}+\upsilon(\beta)p^{\frac{m-2}{2}}\eta((-1)^{\frac
m2}\Delta_0)~~times,
\end{equation}
for each $\beta\in \mathbb{F}_p$, since $\eta((4^m\Delta_0)^{-1})=\eta(\Delta_0)$.

From Lemmas \ref{lm:
exponentialsums} and \ref{lm:type of qua form}, we have
$\eta((-1)^{\frac m2}\Delta_0)=-1$ in this case.

Therefore, by \eqref{eq13:N0} and \eqref{eq13:bi}, we find that
\begin{eqnarray*}
N_{a,b}(0)=\left\{\begin{array}{l}\!p^{m-1}\!-\!(p\!-\!1)p^{\frac{m-2}{2}}\\
~~~~~~~~~occurring~~(p^{m-1}\!-\!(p\!-\!1)p^{\frac{m-2}{2}})
|R_{0,-1}|~~times,\\
\!p^{m-1}\!+\!p^{\frac{m-2}{2}}\\
~~~~~~~~~occurring~~(p-1)(p^{m-1}\!+\!p^{\frac{m-2}{2}})|R_{0,-1}|~times.\\
\end{array}
\right.
\end{eqnarray*}

\emph{Case 3:} $a\in R_1$. In this case, rank$(Q(x))=m-2d$ by Lemma
\ref{lm: exponentialsums}. And consequently we can assume that the
coefficients in \eqref{N1:0 2} satisfy
$\prod^{m-2d}_{i=1}a_i\neq0$ and $a_i=0$ for
$m-2d<i\leqslant m$. Then \eqref{N1:0 2} is equivalent to

$$\sum^{m-2d}_{i=1}a_ix^2_i+\sum^{m}_{i=1}b_ix_i=0.$$

If there exists some $b_i\neq0$ for $m-2d<i\leqslant m,$ we can
assume without loss of generality that $b_m\neq0$. Then
$N_{a,b}(0)=p^{m-1}$, since we can substitute arbitrary elements of
$\mathbb{F}_p$ for $x_1,\cdots,x_{m-1}$ and the value of $x_m$ is
then uniquely determined. Furthermore, there are exactly
$p^m-p^{m-2d}$ choices for $b$ such that there is at least one
$b_i\neq0$ for $m-2d<i\leqslant m,$ as $b$ runs through
$\mathbb{F}_{p^m}$.

 If $b_i=0$ for all  $m-2d<i\leqslant m$,
then the substitution $x_i=y_i-\frac{b_i}{2a_i}$ for $1\leqslant
i\leqslant m-2d$ yields
$$\sum^{m-2d}_{i=1}a_iy^2_i=\sum^{m-2d}_{i=1}\frac{b^2_i}{4a_i}.$$
Notice that $m-2d$
is even. By Lemmas \ref{lm: solution of quadra form}, \ref{lm:
exponentialsums} and \ref{lm:type of qua form}, we obtain
\begin{eqnarray*}
N_{a,b}(0)&=&p^{2d}\left(p^{m-2d-1}+\upsilon\left(\sum^{m-2d}_{i=1}\frac{b^2_i}{4a_i}\right)p^{\frac{m-2d-2}{2}}\eta((-1)^{\frac
{m-2d}2}\Delta_1)\right)\\
&=&p^{m-1}+\upsilon\left(\sum^{m-2d}_{i=1}\frac{b^2_i}{4a_i}\right)p^{\frac{m+2d-2}{2}}\\
&=&\left\{\begin{array}{l}p^{m-1}+(p-1)p^{\frac{m+2d-2}{2}}~\\
~~~~~~~~~occurring~~(p^{m-2d-1}\!+\!(p\!-\!1)p^{\frac{m-2d-2}{2}})|R_{1,1}|~~times,\\
p^{m-1}-p^{\frac{m+2d-2}{2}}\\
~~~~~~~~~occurring~~(p-1)(p^{m-2d-1}\!-\!p^{\frac{m-2d-2}{2}})|R_{1,1}|~~times,\\
\end{array}
\right.
\end{eqnarray*}
since in this case, $\eta((-1)^{\frac {m-2d}2}\Delta_1)=1$.

By the discussion above, we will get
the result for case $\upsilon_2(m)=\upsilon_2(k)+1$ and $m\neq2k$
described in \eqref{W1:31 m not equal 2k}.

Here we only give the frequencies of the codewords with weight
$(p-1)p^{m-1}$ and $(p-1)(p^{m-1}+p^{\frac{m-2}{2}})$. Other cases
can be proved in a similar manner.

The weight of $\mathsf{c}_1(a,b)$
is equal to $(p-1)p^{m-1}$ if and only if $N_{a,b}(0)=p^{m-1}$.
According to the above analysis, the frequency is
\begin{eqnarray*}p^m&-&1+(p^m-p^{m-2d})|R_{1,1}|\\
&=&p^m-1+(p^m-p^{m-2d})\frac{p^m-1}{p^d+1}\\
&=&(p^m-1)(1+p^{m-d}-p^{m-2d}).\end{eqnarray*}

 The weight of $\mathsf{c}_1(a,b)$ is equal to
$(p-1)(p^{m-1}+p^{\frac{m-2}{2}})$ if and only if
$N_{a,b}(0)=p^{m-1}-(p-1)p^{\frac {m-2}2}$. The frequency is equal
to
$$(p^{m-1}-(p-1)p^{\frac{m-2}{2}})|R_{0,-1}|=(p^{m-1}-(p-1)p^{\frac{m-2}{2}})\frac{p^d(p^m-1)}{p^d+1}.$$


\textcircled{2}$\upsilon_2(m)>\upsilon_2(k)+1$.

 The value of $N_{a,b}(0)$
 will be computed by distinguishing among the following cases.

%

\emph{Case 1:} $a=0$. In this case, if $b=0$ then $N_{a,b}(0)=p^m$,
and this value occurs only once, and if $b \neq 0$ then
$N_{a,b}(0)=p^{m-1}$, and this value occurs $p^m-1$ times.

\emph{Case 2:} $a\in R_0$. In this case, rank$(Q(x))=m$ by Lemma
\ref{lm: exponentialsums} and consequently every coefficient $a_i$
in \eqref{N1:0 2} is nonzero.

For $1\leqslant i\leqslant m,$ let $x_i=y_i-\frac{b_i}{2a_i},$ then
\eqref{N1:0 2} is equivalent to $\sum^m_{i=1}
a_iy^2_i=\sum ^m_{i=1}\frac{b^2_i}{4a_i}$. According to Lemma
\ref{lm: solution of quadra form}, we have
\begin{eqnarray}\label{eq14:N0}N_{a,b}(0)=p^{m-1}+\upsilon\left(\sum^m_{i=1}\frac{b^2_i}{4a_i}\right)p^{\frac{m-2}{2}}\eta((-1)^{\frac
m2}\Delta_0).
\end{eqnarray}
Note that $\sum ^m_{i=1}\frac{b^2_i}{4a_i}$ can be regarded as a quadratic form
in $m$ variables $b_i$ for $1\leqslant i\leqslant m$.
Again by Lemma \ref{lm: solution of quadra form}, as $b$ runs through $\mathbb{F}_{p^m}$, we obtain
\begin{equation}\label{eq14:bi}\sum^m_{i=1}\frac{b^2_i}{4a_i}=\beta~~ occurring~~
p^{m-1}+\upsilon(\beta)p^{\frac{m-2}{2}}\eta((-1)^{\frac
m2}\Delta_0)~~ times,
\end{equation} for every $\beta\in \mathbb{F}_p$.

 By Lemmas \ref{lm:
exponentialsums} and \ref{lm:type of qua form}, we have
$\eta((-1)^{\frac m2}\Delta_0)=1$ in this case.

Therefore, combining \eqref{eq14:N0} and \eqref{eq14:bi} gives
\begin{eqnarray*}
N_{a,b}(0)=\left\{\begin{array}{l}p^{m-1}+(p-1)p^{\frac{m-2}{2}}\\
~~~~~~~~~occurring~(p^{m-1}+(p-1)p^{\frac{m-2}{2}})
|R_{0,1}|~~times,\\
p^{m-1}-p^{\frac{m-2}{2}}\\
~~~~~~~~~occurring~~(p-1)(p^{m-1}-p^{\frac{m-2}{2}})|R_{0,1}|~~times.\\
\end{array}
\right.
\end{eqnarray*}

\emph{Case 3:} $a\in R_1$. In this case, rank$(Q(x))=m-2d$ by Lemma
\ref{lm: exponentialsums}. Similarly, suppose that the coefficients
in \eqref{N1:0 2} satisfy $\prod^{m-2d}_{i=1}a_i\neq0$
and $a_i=0$ for $m-2d<i\leqslant m$. Then \eqref{N1:0 2} is
equivalent to

$$\sum^{m-2d}_{i=1}a_ix^2_i+\sum^{m}_{i=1}b_ix_i=0.$$

If there exists some $b_i\neq0$ for $m-2d<i\leqslant m,$ then
$N_{a,b}(0)=p^{m-1}$ and there are exactly $p^m-p^{m-2d}$ choices
for $b$ such that there is at least one $b_i\neq0$ for
$m-2d<i\leqslant m,$ as $b$ runs through $\mathbb{F}_{p^m}$.

If $b_i=0$ for all  $m-2d<i\leqslant m$, then the substitution
$x_i=y_i-\frac{b_i}{2a_i}$ for $1\leqslant i\leqslant m-2d$ yields
$$\sum^{m-2d}_{i=1}a_iy^2_i=\sum^{m-2d}_{i=1}\frac{b^2_i}{4a_i}.$$
It then follows from Lemmas \ref{lm: solution of quadra form}, \ref{lm:
exponentialsums} and \ref{lm:type of qua form} that
\begin{eqnarray*}
N_{a,b}(0)&=&p^{2d}\left(p^{m-2d-1}+\upsilon\left(\sum^{m-2d}_{i=1}\frac{b^2_i}{4a_i}\right)p^{\frac{m-2d-2}{2}}\eta((-1)^{\frac
{m-2d}2}\Delta_1)\right)\\
&=&p^{m-1}-\upsilon\left(\sum^{m-2d}_{i=1}\frac{b^2_i}{4a_i}\right)p^{\frac{m+2d-2}{2}}\\
&=&\left\{\begin{array}{l}p^{m-1}-(p-1)p^{\frac{m+2d-2}{2}}\\
~~~~~~~~~occurring~~(p^{m-2d-1}\!-\!(p\!-\!1)p^{\frac{m-2d-2}{2}})|R_{1,-1}|~~times,\\
p^{m-1}+p^{\frac{m+2d-2}{2}}\\
~~~~~~~~~occurring~~(p\!-\!1)(p^{m-2d-1}\!+\!p^{\frac{m-2d-2}{2}})|R_{1,-1}|~~times,\\
\end{array}
\right.
\end{eqnarray*}
since $\eta((-1)^{\frac {m-2d}2}\Delta_1)=-1$.

Combining all above cases and using Equation \eqref{W1:1234}, we will get
the result for case $\upsilon_2(m)>\upsilon_2(k)+1$ described in
\eqref{W1:4 v2(m)>v2(k)+1}.

Here we give the frequencies of the codewords with weight
$(p-1)p^{m-1}$ and $(p-1)(p^{m-1}-p^{\frac{m-2}{2}})$. Other cases
can be obtained in a similar manner.

The weight of $\mathsf{c}_1(a,b)$
is equal to $(p-1)p^{m-1}$ if and only if $N_{a,b}(0)=p^{m-1}$.
By the above argument, we see that the frequency is
\begin{eqnarray*}p^m&-&1+(p^m-p^{m-2d})|R_{1,-1}|\\
&=&p^m-1+(p^m-p^{m-2d})\frac{p^m-1}{p^d+1}\\
&=&(p^m-1)(1+p^{m-d}-p^{m-2d}).\end{eqnarray*}

 The weight of
$\mathsf{c}_1(a,b)$ is equal to $(p-1)(p^{m-1}-p^{\frac{m-2}{2}})$
if and only if $N_{a,b}(0)=p^{m-1}+(p-1)p^{\frac {m-2}2}$. The
frequency is equal to
$$(p^{m-1}+(p-1)p^{\frac{m-2}{2}})|R_{0,1}|=(p^{m-1}+(p-1)p^{\frac{m-2}{2}})\frac{p^d(p^m-1)}{p^d+1}.$$

%
This completes the whole proof of Theorem \ref{thm:code a bx}.
\hfill\space$\qed$
\end{proof}

\begin{corollary}\label{coro:C1}

If $m=2k$, then $C_1$ is a cyclic code over $\mathbb{F}_p$ with
parameters $[p^m-1,3m/2]$ and the weight distribution is
given as follows:
\begin{eqnarray}\label{W1:32 m=2k}
\left\{\begin{array}{l}A_0=1,\\A_{(p-1)p^{m-1}}=p^m-1,\\
  A_{(p-1)(p^{m-1}+p^{\frac{m-2}{2}})}=(p^{m-1}-(p-1)p^{\frac{m-2}{2}})(p^{\frac{m}{2}}-1), \\
  A_{(p-1)p^{m-1}-p^{\frac{m-2}{2}}}=(p-1)(p^{m-1}+p^{\frac{m-2}{2}})(p^{\frac{m}{2}}-1).\\
\end{array}
\right.
\end{eqnarray}
\end{corollary}
\begin{proof}

Let $K=\{x\in\mathbb{F}_{p^m}\big|~x^{p^k}+x=0\}$.
 It is easy to check that $\mathsf{c}_1(a,b)=\mathsf{c}_1(a+\delta,b)$ for any $\delta\in K$ and $\mathsf{c}_1(a,b)\in C_1$. Hence, $C_1$ is degenerate with
 dimension $3m/2$ over $\mathbb{F}_p$.

 Note that $|K|=p^{\frac{m}{2}}$ and in this case $\upsilon_2(m)=\upsilon_2(k)+1$. Substituting $d=m/2$ to
 Equation \eqref{W1:31 m not equal 2k} and dividing each $A_i$ by
 $p^{\frac{m}{2}}$, we get the result given in \eqref{W1:32 m=2k}.
This finishes the proof of Corollary~\ref{coro:C1}. \hfill\space$\qed$
\end{proof}

%

\begin{remark}
It should be noted that, for $s$ being even, the weight distribution of the code $C_1$ is determined
by Theorem \ref{thm:code a bx} and Corollary
\ref{coro:C1}. The results show that $C_1$ is a cyclic code with three or five weights.
\end{remark}

We give some examples for the code $C_1$ in the case of
$\upsilon_2(m)\geqslant\upsilon_2(k)+1$,
i.e., $s$ is even, which is not included in
\cite{feng2007value,lichao2009covering,yuan2006weight}.

\begin{example}
Let $m=6,k=1,p=3$. This corresponds to the case
$\upsilon_2(m)=\upsilon_2(k)+1$ and $m\neq2k$. Using Magma, $C_1$ is a [728, 12,
432] cyclic linear code over $\mathbb{F}_3$ with the weight
distribution:
\begin{eqnarray*}
&&A_0=1,A_{432}=6006,A_{477}=275184,A_{486}=118664,\\&&A_{504}=122850,A_{513}=8736,
\end{eqnarray*}
which verifies the result of
Equation \eqref{W1:31 m not equal 2k} in Theorem \ref{thm:code a bx}.
\end{example}
\begin{example}
Let $m=4,k=1,p=5$. This corresponds to the case
$\upsilon_2(m)>\upsilon_2(k)+1$. Using Magma, $C_1$ is a [624, 8,
475] cyclic linear code over $\mathbb{F}_5$ with the weight
distribution:
\begin{eqnarray*}
&&A_0=1,A_{475}=2496,A_{480}=75400,A_{500}=63024,\\&&
A_{505}=249600,A_{600}=104,
\end{eqnarray*}which verifies the result of Equation \eqref{W1:4 v2(m)>v2(k)+1} in Theorem
\ref{thm:code a bx}.
\end{example}

\section{The weight distribution of the code $C_2$}\label{sec:C2}
In this section, we will study the weight distribution of the code $C_2$ as
described in Section \ref{sec:intro}. By the well-known Delsarte's
Theorem \cite{delsarte1975subfield}, we have
$$C_2=\{\mathsf{c}_2(a,c):a,c\in \mathbb{F}_{p^m}\},$$ where $\mathsf{c}_2(a,c)=(\text{Tr}(ax^{p^k+1}+c))_{x\in
\mathbb{F}^*_{p^m}}$.

For any two codewords $\mathsf{c}_2(a_1,c_1)$ and
$\mathsf{c}_2(a_2,c_2)$ in $C_2$ given above, it is easy to verify
that $\mathsf{c}_2(a_1,c_1)=\mathsf{c}_2(a_2,c_2)$ if and only if
$a_1=a_2$ and $\text{Tr}(c_1)=\text{Tr}(c_2)$. Hence, $C_2$ can be expressed as
$$C_2=\{\mathsf{c}_2(a,\lambda)=(\text{Tr}(ax^{p^k+1})-\lambda)_{x\in
\mathbb{F}^*_{p^m}}:a\in
\mathbb{F}_{p^m},\lambda\in\mathbb{F}_{p}\},$$ where
$\lambda=-\text{Tr}(c)$.

Let $N_{a,\lambda}(0)$ be the number of solutions $x\in
\mathbb{F}_{p^m}$ satisfying
\begin{eqnarray}\label{N2:0}
 \text{Tr}(ax^{p^k+1})-\lambda=0,
\end{eqnarray}
as $(a,\lambda)$ runs through
$\mathbb{F}_{p^m}\times\mathbb{F}_{p}$. By making a nonsingular
linear substitution as introduced in Section \ref{sec:Preli}, Equation
\eqref{N2:0} is equivalent to
\begin{eqnarray}\label{N2:0 2}
 \sum^m_{i=1}a_ix^2_i=\lambda,
\end{eqnarray} where $a_i\in\mathbb{F}_{p}$. Thus, $N_{a,\lambda}(0)$ also represents the number
of $(x_1,x_2,\dots,x_m)\in \mathbb{F}^m_{p}$ satisfying \eqref{N2:0 2}.

In the following, we establish the weight distribution of the code $C_2$
when $(a,\lambda)$ runs through
$\mathbb{F}_{p^m}\times\mathbb{F}_{p}$.

\begin{theorem}\label{thm:code a c}With notation as above. If $m\neq2k$, then $C_2$ is a cyclic code over $\mathbb{F}_p$
with parameters $[p^m-1,m+1]$ and\\
\textcircled{1} If $0=\upsilon_2(m)\leqslant \upsilon_2(k)$, the
weight distribution of $C_2$ is given as follows:
\begin{eqnarray}\label{W2:1 v2(m)=0}
\left\{\begin{array}{l}A_0=1,\\A_{p^m-1}=p-1,\\
  A_{(p-1)p^{m-1}}=p^m-1,\\
  A_{(p-1)p^{m-1}-p^{\frac{m-1}2}-1}=\frac{p-1}{2}(p^m-1),\\
  A_{(p-1)p^{m-1}+p^{\frac{m-1}2}-1}=\frac{p-1}{2}(p^m-1).\\
\end{array}
\right.
\end{eqnarray}
\textcircled{2} If $1\leqslant \upsilon_2(m)\leqslant
\upsilon_2(k)$, the weight distribution of $C_2$ is given as
follows:
\begin{eqnarray}\label{W2:2 1<=v2(m)}
\left\{\begin{array}{l}A_0=1,\\A_{p^m-1}=p-1,\\
  A_{(p-1)p^{m-1}-p^{\frac{m-2}2}-1}=\frac{p-1}{2}(p^m-1),\\
  A_{(p-1)p^{m-1}+p^{\frac{m-2}2}-1}=\frac{p-1}{2}(p^m-1), \\
  A_{(p-1)(p^{m-1}-p^{\frac{m-2}2})}=\frac{1}{2}(p^m-1), \\
  A_{(p-1)(p^{m-1}+p^{\frac{m-2}2})}=\frac{1}{2}(p^m-1).
\end{array}
\right.
\end{eqnarray}
\textcircled{3} If $v_2(m)=v_2(k)+1$, the weight distribution of
$C_2$  is given as follows:
\begin{eqnarray}\label{W2:31 m not equal 2k}
\left\{\begin{array}{l}A_0=1,\\A_{p^m-1}=p-1,\\
  A_{(p-1)(p^{m-1}+p^{\frac{m-2}2})}=\frac{p^d(p^m-1)}{p^d+1},\\
    A_{(p-1)p^{m-1}-p^{\frac{m-2}2}-1}=\frac{p^d(p-1)(p^m-1)}{p^d+1},\\
    A_{(p-1)(p^{m-1}-p^{\frac{m+2d-2}2})}=\frac{p^m-1}{p^d+1},\\
   A_{(p-1)p^{m-1}+p^{\frac{m+2d-2}2}-1}=\frac{(p-1)(p^m-1)}{p^d+1}.
\end{array}
\right.
\end{eqnarray}
\textcircled{4} If $\upsilon_2(m)>\upsilon_2(k)+1$, the weight
distribution of $C_2$ is given as follows:
\begin{eqnarray}\label{W2:4 v2(m)>v2(k)+1}
\left\{\begin{array}{l}A_0=1,\\A_{p^m-1}=p-1,\\
  A_{(p-1)(p^{m-1}-p^{\frac{m-2}2})}=\frac{p^d(p^m-1)}{p^d+1},\\
   A_{(p-1)p^{m-1}+p^{\frac{m-2}2}-1}=\frac{p^d(p-1)(p^m-1)}{p^d+1},\\
   A_{(p-1)(p^{m-1}+p^{\frac{m+2d-2}2})}=\frac{(p^m-1)}{p^d+1},\\
   A_{(p-1)p^{m-1}-p^{\frac{m+2d-2}2}-1}=\frac{(p-1)(p^m-1)}{p^d+1}.
\end{array}
\right.
\end{eqnarray}
\end{theorem}

\begin{proof}
The length and dimension follow immediately from the
definition of the code $C_2$. The Hamming weight of every codeword
$\mathsf{c}_2(a,\lambda)$ can be determined by
\begin{eqnarray}\label{W2:1234}
\omega t(\mathsf{c}_2(a,\lambda)) &=&p^m-1-\#\{x\in
\mathbb{F}^*_{p^m}\big|~\text{Tr}(ax^{p^k+1})-\lambda=0\}\nonumber\\
&=&\left\{\begin{array}{ll}p^m-N_{a,\lambda}(0), &if~~ \lambda=0,\\
p^m-1-N_{a,\lambda}(0), &if ~~\lambda\neq 0, \\
\end{array}
\right.
\end{eqnarray}where $\lambda=-\text{Tr}(c)$.

We will calculate the weight distribution of the code $C_2$ by distinguishing the following
cases.

 \textcircled{1} $0=\upsilon_2(m)\leqslant \upsilon_2(k)$.

 The value of $N_{a,\lambda}(0)$
 will be computed according to the choice of the parameter $a$.

\emph{Case 1:} $a=0$. In this case, if $\lambda=0$ then
$N_{a,\lambda}(0)=p^m$, and this value occurs only once, and if
$\lambda \neq 0$ then $N_{a,\lambda}(0)=0$, and this value occurs
$p-1$ times.

\emph{Case 2:} $a\in  \mathbb{F}^*_{p^m}$, i.e., $a\in R_0$. In this case,
rank$(Q(x))=m$ by Lemma \ref{lm: exponentialsums} and consequently
every coefficient $a_i$ in \eqref{N2:0 2} is nonzero.

From Lemma \ref{lm: solution of quadra form}, we have
$$
N_{a,\lambda}(0)=
p^{m-1}+p^{\frac{m-1}2}\eta({(-1)}^{\frac{m-1}2}\lambda\Delta_0).
$$

If $\lambda=0$ then $N_{a,\lambda}(0)=p^{m-1}$, and this value
occurs $p^m-1$ times.

If $\lambda \neq 0$, then there are $(p-1)/2$ squares and
nonsquares in $\mathbb{F}^*_p$, respectively. If $\lambda$ is a
square in $\mathbb{F}^*_p$, then
$$N_{a,\lambda}(0)=p^{m-1}+
p^{\frac{m-1}2}\eta((-1)^{\frac{m-1}2}\Delta_0).$$ Using Lemma
\ref{lm: exponentialsums} and Lemma \ref{lm:type of qua form}, we
find that
\begin{equation*}
N_{a,\lambda}(0)=\left\{\begin{array}{ll}p^{m-1}+ p^{\frac{m-1}2}&~~occurring~~ \frac{p-1}2|R_{0,1}|~~~~times,\\
p^{m-1}- p^{\frac{m-1}2} &~~occurring~~
\frac{p-1}2|R_{0,-1}| ~~times.\\
\end{array}\right.
\end{equation*}

Similarly, if $\lambda$ is a nonsquare in $\mathbb{F}^*_p$, then
 $$N_{a,\lambda}(0)=p^{m-1}-
p^{\frac{m-1}2}\eta((-1)^{\frac{m-1}2}\Delta_0).$$
This leads to
\begin{equation*}
N_{a,\lambda}(0)=\left\{\begin{array}{ll}p^{m-1}- p^{\frac{m-1}2}&~~occurring~~ \frac{p-1}2|R_{0,1}|~~~~times,\\
p^{m-1}+ p^{\frac{m-1}2} &~~occurring~~
\frac{p-1}2|R_{0,-1}| ~~times.\\
\end{array}\right.
\end{equation*}

By Equation \eqref{W2:1234} and the above analysis, we will derive the
result for case $0=\upsilon_2(m)\leqslant \upsilon_2(k)$ described
in \eqref{W2:1 v2(m)=0}.

Here we give the frequencies of the codewords with weight
$(p-1)p^{m-1}$ and $(p-1)p^{m-1}- p^{\frac {m-1}2}-1$. Other cases
can be analyzed in a similar way.

The weight of
$\mathsf{c}_2(a,\lambda)$ is equal to $(p-1)p^{m-1}$ if and only if
$N_{a,\lambda}(0)=p^{m-1}$ and $\lambda=0$. Thus the
above argument shows that the frequency is $p^m-1.$

The weight of $\mathsf{c}_2(a,\lambda)$ is equal to $(p-1)p^{m-1}-
p^{\frac {m-1}2}-1$ if and only if
$N_{a,\lambda}(0)=p^{m-1}+p^{\frac {m-1}2}$ and $\lambda\neq0$. The
frequency is equal to
$$\frac{p-1}{2}(|R_{0,1}|+|R_{0,-1}|)=\frac{p-1}{2}(p^m-1).$$

%

\textcircled{2} $1\leqslant \upsilon_2(m)\leqslant \upsilon_2(k).$

 The value of $N_{a,\lambda}(0)$
 will be calculated by distinguishing the case $a=0$ from the case $a\neq0$.

\emph{Case 1:} $a=0$. In this case, if $\lambda=0$ then
$N_{a,\lambda}(0)=p^m$, and this value occurs only once, and if
$\lambda \neq 0$ then $N_{a,\lambda}(0)=0$, and this value occurs
$p-1$ times.

\emph{Case 2:} $a\in  \mathbb{F}^*_{p^m}$, i.e., $a\in R_0$. In this case,
rank$(Q(x))= m$ by Lemma \ref{lm: exponentialsums} and consequently
every coefficient $a_i$ in \eqref{N2:0 2} is nonzero.

Applying Lemma \ref{lm: solution of quadra form} gives that
$$
N_{a,\lambda}(0) =
p^{m-1}+\upsilon(\lambda)p^{\frac{m-2}2}\eta({(-1)}^{\frac{m}2}\Delta_0).
$$

If $\lambda=0$ then
$$N_{a,\lambda}(0)=p^{m-1}+(p-1)p^{\frac{m-2}2}\eta({(-1)}^{\frac{m}2}\Delta_0).$$
It then follows from Lemmas \ref{lm: exponentialsums} and \ref{lm:type of qua
form} that
\begin{equation*}
N_{a,\lambda}(0)=\left\{\begin{array}{ll}p^{m-1}+(p-1)p^{\frac{m-2}2}&~~occurring~~ |R_{0,1}|~~~~times,\\
p^{m-1}-(p-1)p^{\frac{m-2}2} &~~occurring~~
|R_{0,-1}| ~~times.\\
\end{array}\right.
\end{equation*}

If $\lambda \neq 0$ then
$$N_{a,\lambda}(0)=p^{m-1}-p^{\frac{m-2}2}\eta({(-1)}^{\frac{m}2}\Delta_0).$$
Again by Lemmas \ref{lm: exponentialsums} and \ref{lm:type of
qua form}, we have
\begin{equation*}
N_{a,\lambda}(0)=\left\{\begin{array}{ll}p^{m-1}- p^{\frac{m-2}2}&~~occurring~~ (p-1)|R_{0,1}|~~~~times,\\
p^{m-1}+p^{\frac{m-2}2} &~~occurring~~
(p-1)|R_{0,-1}| ~~times.\\
\end{array}\right.
\end{equation*}

By Equation \eqref{W2:1234} and the above analysis, we will get the
result for case  $1\leqslant \upsilon_2(m)\leqslant \upsilon_2(k)$
described in \eqref{W2:2 1<=v2(m)}.

Here we give the frequencies of the codewords with weight
$(p-1)(p^{m-1}-p^{\frac{m-2}2})$ and
$(p-1)p^{m-1}-p^{\frac{m-2}2}-1$. Other cases can be similarly verified.

The weight of $\mathsf{c}_2(a,\lambda)$ is equal to
$(p-1)(p^{m-1}-p^{\frac{m-2}2})$ if and only if
$N_{a,\lambda}(0)=p^{m-1}+(p-1)p^{\frac{m-2}2}$ and $\lambda=0$.
Based on the above discussion, the frequency is $|R_{0,1}|=\frac{p^m-1}{2}.$

The weight of $\mathsf{c}_2(a,\lambda)$ is equal to $(p-1)p^{m-1}-
p^{\frac {m-2}2}-1$ if and only if
$N_{a,\lambda}(0)=p^{m-1}+p^{\frac {m-2}2}$ and $\lambda\neq0$. Therefore, the
frequency is equal to $$(p-1)|R_{0,-1}|=\frac{1}{2}(p-1)(p^m-1).$$

%

\textcircled{3}Let $v_2(m)=v_2(k)+1$ and $m\neq2k$.

 The value of $N_{a,\lambda}(0)$
 will be calculated by distinguishing among the following cases.

 \emph{Case 1:} $a=0$. In this case, if $\lambda=0$ then $N_{a,\lambda}(0)=p^m$, and this
value occurs only once, and if $\lambda \neq 0$ then
$N_{a,\lambda}(0)=0$, and this value occurs $p-1$ times.

\emph{Case 2:} $a\in R_0$. In this case, rank$(Q(x))=m$ and consequently every coefficient $a_i$
in \eqref{N2:0 2} is nonzero.

From Lemma \ref{lm: solution of quadra form}, we have
$$N_{a,\lambda}(0)=p^{m-1}+\upsilon(\lambda)p^{\frac{m-2}{2}}\eta((-1)^{\frac
m2}\Delta_0).$$

It then follows from Lemmas \ref{lm:
exponentialsums} and \ref{lm:type of qua form} that
$$N_{a,\lambda}(0)=p^{m-1}-\upsilon(\lambda)p^{\frac{m-2}{2}},$$
since $\eta((-1)^{\frac m2}\Delta_0)=-1$.

If $\lambda=0$, then $N_{a,\lambda}(0)=p^{m-1}-(p-1)p^{\frac{m-2}{2}}$ occurring
$|R_{0,-1}|$ times.

 If $\lambda\neq0$, then
$N_{a,\lambda}(0)=p^{m-1}+p^{\frac{m-2}{2}}$ occurring $(p-1)|R_{0,-1}|$ times.

\emph{Case 3:} $a\in R_1$.  In this case, rank$(Q(x))=m-2d$. Again by Lemmas \ref{lm: solution of quadra form}, \ref{lm: exponentialsums} and \ref{lm:type of qua form}, we
find
\begin{eqnarray*}N_{a,\lambda}(0)&=&p^{2d}(p^{m-2d-1}+\upsilon(\lambda)p^{\frac{m-2d-2}{2}}\eta((-1)^{\frac
{m-2d}2}\Delta_1)\\
&=&p^{m-1}+\upsilon(\lambda)p^{\frac{m+2d-2}{2}},
\end{eqnarray*}since $\eta((-1)^{\frac
{m-2d}2}\Delta_1)=1$.

If $\lambda=0$, then $N_{a,\lambda}(0)=p^{m-1}+(p-1)p^{\frac{m+2d-2}{2}}$
occurring $|R_{1,1}|$ times.

 If $\lambda\neq0$, then
$N_{a,\lambda}(0)=p^{m-1}-p^{\frac{m+2d-2}{2}}$ occurring $(p-1)|R_{1,1}|$
times.

By Equation \eqref{W2:1234} and the above analysis, we will obtain the
result for case  $v_2(m)=v_2(k)+1$ and $m\neq2k$ described in
\eqref{W2:31 m not equal 2k}.

Here we give the frequencies of the codewords with weight $p^{m}-1$
and $(p-1)(p^{m-1}+p^{\frac{m-2}2})$. Other cases can be analyzed in
an analogous manner.

 The weight of $\mathsf{c}_2(a,\lambda)$ is equal to
$p^{m}-1$ if and only if $N_{a,\lambda}(0)=0$ and $\lambda\neq0$.
The above discussion shows that the frequency is $p-1.$

The weight of $\mathsf{c}_2(a,\lambda)$ is equal to
$(p-1)(p^{m-1}+p^{\frac{m-2}2})$ if and only if
$N_{a,\lambda}(0)=p^{m-1}-(p-1)p^{\frac {m-2}2}$ and $\lambda=0$.
The frequency is $|R_{0,-1}|=\frac{p^d(p^m-1)}{p^d+1}.$

%

\textcircled{4}Let $\upsilon_2(m)>\upsilon_2(k)+1$.

 The value of $N_{a,\lambda}(0)$
 will be calculated according to the choice of the parameter $a$.

 \emph{Case 1:} $a=0$. In this case, if $\lambda=0$ then $N_{a,\lambda}(0)=p^m$, and this
value occurs only once, and if $\lambda \neq 0$ then
$N_{a,\lambda}(0)=0$, and this value occurs $p-1$ times.

\emph{Case 2:} $a\in R_0$. In this case, rank$(Q(x))=m$ and consequently every coefficient $a_i$
in \eqref{N2:0 2} is nonzero.

It then follows from Lemma \ref{lm: solution of quadra form} that
$$N_{a,\lambda}(0)=p^{m-1}+\upsilon(\lambda)p^{\frac{m-2}{2}}\eta((-1)^{\frac
m2}\Delta_0).$$

Applying Lemmas \ref{lm:
exponentialsums} and \ref{lm:type of qua form} yields that
$$N_{a,\lambda}(0)=p^{m-1}+\upsilon(\lambda)p^{\frac{m-2}{2}},$$
since $\eta((-1)^{\frac
m2}\Delta_0)=1$.

If $\lambda=0$, then $N_{a,\lambda}(0)=p^{m-1}+(p-1)p^{\frac{m-2}{2}}$ occurring
$|R_{0,1}|$ times.

 If $\lambda\neq0$, then
$N_{a,\lambda}(0)=p^{m-1}-p^{\frac{m-2}{2}}$ occurring $(p-1)|R_{0,1}|$ times.

\emph{Case 3:} $a\in R_1$.  In this case, rank$(Q(x))=m-2d$. Again by Lemmas \ref{lm: solution of quadra form},
\ref{lm: exponentialsums} and \ref{lm:type of qua form}, we
arrive at
\begin{eqnarray*}N_{a,\lambda}(0)&=&p^{2d}(p^{m-2d-1}+\upsilon(\lambda)p^{\frac{m-2d-2}{2}}\eta((-1)^{\frac
{m-2d}2}\Delta_1)\\
&=&p^{m-1}-\upsilon(\lambda)p^{\frac{m+2d-2}{2}},
\end{eqnarray*}
since $\eta((-1)^{\frac
{m-2d}2}\Delta_1)=-1$.

If $\lambda=0$, then $N_{a,\lambda}(0)=p^{m-1}-(p-1)p^{\frac{m+2d-2}{2}}$
occurring $|R_{1,-1}|$ times.

 If $\lambda\neq0$, then
$N_{a,\lambda}(0)=p^{m-1}+p^{\frac{m+2d-2}{2}}$ occurring $(p-1)|R_{1,-1}|$
times.

By Equation \eqref{W2:1234} and the above analysis, we will derive the
result for case  $\upsilon_2(m)>\upsilon_2(k)+1$ described in
\eqref{W2:4 v2(m)>v2(k)+1}.

Here we only show the frequencies of the codewords with weight $p^{m}-1$
and $(p-1)(p^{m-1}-p^{\frac{m-2}2})$. Other cases are similarly verified.

 The weight of $\mathsf{c}_2(a,\lambda)$ is equal to
$p^{m}-1$ if and only if $N_{a,\lambda}(0)=0$ and $\lambda\neq0$.
From the above discussion, the frequency is $p-1.$

The weight of $\mathsf{c}_2(a,\lambda)$ is equal to
$(p-1)(p^{m-1}-p^{\frac{m-2}2})$ if and only if
$N_{a,\lambda}(0)=p^{m-1}+(p-1)p^{\frac {m-2}2}$ and $\lambda=0$.
The frequency is $|R_{0,1}|=\frac{p^d(p^m-1)}{p^d+1}.$

This completes the proof of this theorem. \hfill\space$\qed$
\end{proof}


\begin{corollary}\label{coro:C2}
If $m=2k$, then $C_2$ is a cyclic code over $\mathbb{F}_p$ with
parameters $[p^m-1,{m}/2+1]$ and the weight distribution is
given as follows:
\begin{eqnarray}\label{W2:32 m=2k}
\left\{\begin{array}{l}A_0=1,\\A_{p^m-1}=p-1,\\
  A_{(p-1)(p^{m-1}+p^{\frac{m-2}2})}=p^{\frac{m}{2}}-1,\\
  A_{(p-1)p^{m-1}-p^{\frac{m-2}2}-1}=(p-1)(p^{\frac{m}{2}}-1).
\end{array}
\right.
\end{eqnarray}
\end{corollary}

\begin{proof} Let
$K=\{x\in\mathbb{F}_{p^m}\big|~x^{p^k}+x=0\}$.
 It is easily checked that $\mathsf{c}_2(a,\lambda)=\mathsf{c}_2(a+\delta,c)$ for any $\delta\in K$ and $\mathsf{c}_2(a,\lambda)\in C_2$. Hence, $C_2$ is degenerate with
 dimension ${m}/{2}+1$ over $\mathbb{F}_p$.

 Note that $|K|=p^{\frac{m}{2}}$ and in this case $\upsilon_2(m)=\upsilon_2(k)+1$. Substituting $d={m}/{2}$ to
 Equation \eqref{W2:31 m not equal 2k} and dividing each $A_i$ by
 $p^{\frac{m}{2}}$, we get the desired result. Now the proof of Corollary \ref{coro:C2}
is complete. \hfill\space$\qed$\end{proof}
%

The following are some examples for the code $C_2$. Note that the weight
distribution of $C_2$ is not known before.
\begin{example}
Let $m=6,k=2,p=3$. This corresponds to the case
$1\leqslant\upsilon_2(m)\leqslant\upsilon_2(k)$. Using Magma, $C_2$
is a [728, 7, 468] cyclic linear code over $\mathbb{F}_3$ with the
weight distribution:
$$A_0=1,A_{468}=364,A_{476}=728,A_{494}=728,A_{504}=364,A_{728}=2,$$
 which confirms the result of Equation \eqref{W2:2 1<=v2(m)} in Theorem \ref{thm:code a c}.
\end{example}
\begin{example}
Let $m=8,k=1,p=3$. This corresponds to the case
$\upsilon_2(m)>\upsilon_2(k)+1$. Using Magma, $C_2$ is a [6560, 9,
4292] cyclic linear code over $\mathbb{F}_3$ with the weight
distribution:
\begin{eqnarray*}
&&A_0=1,A_{4292}=3280,A_{4320}=4920,A_{4400}=9840,\\
&&A_{4536}=1640,A_{6560}=2,
\end{eqnarray*}
 which confirms the result of Equation \eqref{W2:4
 v2(m)>v2(k)+1} in Theorem \ref{thm:code a c} .
\end{example}
\begin{example}
Let $m=6,k=3,p=3$. This corresponds to the case $m=2k$. Using Magma, $C_2$
is a [728, 4, 476] cyclic linear code over $\mathbb{F}_3$ with the
weight distribution:
$$A_0=1,A_{476}=52,A_{504}=26,A_{728}=2,$$
 which confirms the result of Equation \eqref{W2:32
 m=2k} in Corollary \ref{coro:C2}.
\end{example}
\section{Conclusion and remarks}\label{sec:conclusion}

In this paper, we completely determined the  weight distributions of
two classes of cyclic codes $C_1$ for even $s$ and $C_2$ over $\mathbb{F}_p$. The
result showed that they have only few weights. In addition, one can get
the value distributions of the corresponding exponential sums of
$C_1$ and $C_2$ by the method described in the proofs of Theorems \ref{thm:code a
bx} and \ref{thm:code a c} though we did not list them
here.

We mention that the weight distributions of several other cyclic codes
may be solved essentially, such as, a family of $p$-ary cyclic
codes with parity-check polynomial $(x-1)h_1(x)h_2(x)$, where
$h_1(x)$ and $h_2(x)$ are defined in Section \ref{sec:intro}. We
leave this for future work.

\begin{acknowledgements}
The work of Zheng-An Yao is partially supported by the NNSFC (Grant No.11271381), the NNSFC (Grant No.11431015)
and China 973 Program (Grant No. 2011CB808000).
The work of Chang-An Zhao is partially supported by the
NNSFC (Grant No. 61472457).
\end{acknowledgements}

\end{document}